\documentclass [twoside,reqno,draft,12pt] {amsart}

\usepackage{amsfonts}
\usepackage{amssymb}
\usepackage{a4}

\usepackage{color}

\newtheorem{thm}{Theorem}[section]

\newtheorem{lem}[thm]{Lemma}
\newtheorem{prop}[thm]{Proposition}

\theoremstyle{definition}

\theoremstyle{remark}
\newtheorem{rem}{Remark}[section]

\numberwithin{equation}{section}

\newcommand{\C}{\mathbb{C}}

\newcommand{\N}{\mathbb{N}}

\newcommand{\R}{\mathbb{R}}

\newcommand{\supp}{\operatorname{supp}}

\def\hat{\widehat}
\def\tilde{\widetilde}
\def \bfo {\begin {eqnarray*} }
\def \efo {\end {eqnarray*} }
\def \ba {\begin {eqnarray*} }
\def \ea {\end {eqnarray*} }
\def \beq {\begin {eqnarray}}
\def \eeq {\end {eqnarray}}
\def \supp {\hbox{supp }}

\def \det {\hbox{det}}

\def \p {\partial}

\def\hat{\widehat}
\def\tilde{\widetilde}
\def \bfo {\begin {eqnarray*} }
\def \efo {\end {eqnarray*} }
\def \ba {\begin {eqnarray*} }
\def \ea {\end {eqnarray*} }
\def \beq {\begin {eqnarray}}
\def \eeq {\end {eqnarray}}
\def \supp {\hbox{supp }}

\def \det {\hbox{det}}

\def \p {\partial}

\begin{document}

 \title[Transmission eigenvalues]{Transmission eigenvalues  for operators with constant coefficients}

\author[Hitrik]{Michael Hitrik}

\address
        {M. Hitrik,  Department of Mathematics\\ 
    UCLA\\ 
    Los Angeles\\ 
    CA 90095-1555\\ 
    USA }

\email{hitrik@math.ucla.edu}

\author[Krupchyk]{Katsiaryna Krupchyk}

\address
        {K. Krupchyk, Department of Mathematics and Statistics \\
         University of Helsinki\\
         P.O. Box 68 \\
         FI-00014   Helsinki\\
         Finland}

\email{katya.krupchyk@helsinki.fi}

\author[Ola]{Petri Ola}

\address
        {P. Ola, Department of Mathematics and Statistics \\
         University of Helsinki\\
         P.O. Box 68 \\
         FI-00014   Helsinki\\
         Finland}

\email{Petri.Ola@helsinki.fi}

\author[P\"aiv\"arinta]{Lassi P\"aiv\"arinta}

\address
        {L. P\"aiv\"arinta, Department of Mathematics and Statistics \\
         University of Helsinki\\
         P.O. Box 68 \\
         FI-00014   Helsinki\\
         Finland}

\email{Lassi.Paivarinta@rni.helsinki.fi}

\maketitle

\begin{abstract} In this paper we study the interior transmission problem and transmission eigenvalues for multiplicative perturbations of linear partial differential operator of order $\ge 2$ with constant real coefficients. 
Under  suitable  growth conditions on the symbol of the operator and the perturbation, we show  the discreteness of the set of transmission eigenvalues and derive sufficient conditions on the existence of transmission eigenvalues. We apply these techniques to the case of the biharmonic operator and the Dirac system. In the hypoelliptic case we present a connection to scattering  theory.

\end{abstract}

\section{Introduction}

Let $V\in L^\infty(\R^n)$ be  compactly supported  with  $\supp(V)=\overline{\Omega}$, where $\Omega\subset \R^n$ is a bounded domain, and let  $P_0(D)$ be a partial differential operator  of order $m\ge 2$ with constant real coefficients


\[
P_0(D)=\sum_{|\alpha|\le m} a_{\alpha}D^\alpha, \quad a_\alpha\in\R, \quad D_j=-i\frac{\partial}{\partial x_j},\quad j=1,\dots,n.
\] 
For any $\varphi\in C^\infty_0(\Omega)$, we set
\begin{equation}
\label{eq_norm}
\|\varphi\|_{P_0}=\|P_0(D)\varphi\|.
\end{equation}
Here and in what follows the notation $\|\cdot\|$
stands for the standard $L^2$-norm. The completion of $C^\infty_0(\Omega)$ with respect to the norm \eqref{eq_norm} is denoted by $H^{P_0}_0(\Omega)$. 

The {\em interior transmission problem} is the following degenerate boundary value problem,
\begin{equation}
\label{eq_ITP_johd}
\begin{aligned}
(P_0-\lambda)v=0 \quad &\text{in} \quad \Omega,\\
(P_0+V-\lambda)w=0 \quad &\text{in} \quad \Omega,\\
 v-w \in H^{P_0}_0(\Omega).
\end{aligned}
\end{equation}

We say that $\lambda\in \C$ is a {\em transmission eigenvalue} if the problem \eqref{eq_ITP_johd} has non-trivial solutions $0\ne v\in L^2_{\textrm{loc}}$ and  $0\ne w\in L^2_{\textrm{loc}}$.

The problem \eqref{eq_ITP_johd} arises naturally in the study of inverse scattering theory for the operator $P_0+V$, when $P_0=-\Delta$. In this context,  the significance of the notion of a transmission eigenvalue is twofold. On the one hand, in reconstruction algorithms of inverse scattering theory, such as the linear sampling method of Colton and Kirsch \cite{CakColbook, ColKir96}, and the factorization method of Kirsch  \cite{KirGribook}, transmission eigenvalues correspond to frequencies that one needs to avoid. The fact that they form a discrete set makes this possible. On the other hand, the knowledge of the transmission eigenvalues carries information about the scatterer \cite{paisyl08}, and in the radially symmetric case, they determine the potential completely 
\cite{CakColGint_complex}, see also 
\cite{McLP} for the previous work in this direction. It is therefore natural to investigate 
the question of existence of transmission eigenvalues, and relate their properties to the properties of the scatterer.

The interior transmission problem was first introduced in 1988 by Colton and Monk \cite{ColMonk88} in connection with an inverse scattering problem for the reduced wave equation. They were led to this problem when studying the injectivity of the far field operator.  The fact that the interior transmission eigenvalues form a discrete set was shown in the Helmholz case in \cite{vanhatkonnatI}.  The problem of existence of transmission eigenvalues, however, remained unsolved until recently. When \(P_0 = -\Delta\),  P\"aiv\"arinta and Sylvester in \cite{paisyl08} proved the first existence result, and soon 
thereafter, the existence of an infinite set of real transmission eigenvalues was established by  Cakoni, Gintides, and Haddar in  \cite{CakDroHou}. 
We would also like to mention recent results on transmission eigenvalues for Maxwell's equations, as well as for the Helmholtz equation, in the presence
of cavities \cite{CakColHous10, CakKir, kir07}.

In this paper we consider the case of quite general operators $P_0$ with constant coefficients, for which a well developed scattering theory is available 
\cite{horbookII}. In particular, we can allow non-elliptic operators as well as operators of higher order. We then show the discreteness of the set of transmission eigenvalues, and derive sufficient conditions for the existence of an infinite set of real transmission eigenvalues. The extension to the case of general operators with constant coefficients seems natural, as in applications one frequently encounters such operators which could be either non-elliptic, or elliptic of higher order.  We also illustrate the techniques developed by applying them in two cases of physically significant operators, namely the biharmonic operator and the Dirac system  in $\R^3$. 

Finally, we point out that in view of the non-selfadjoint nature of the problem  \eqref{eq_ITP_johd},
a study of the distribution of transmission eigenvalues in the complex plane would be both interesting and natural. So far, only in the case when $P_0=-\Delta$, the existence of complex transmission eigenvalues has been shown in the very recent paper \cite{CakColGint_complex}, assuming that the index of refraction is constant and sufficiently close to one. See also the numerical computation of complex transmission eigenvalues given in \cite{ColMonkSun}.

The structure of the paper is as follows. In Section 2 we formulate the problem and state a compact embedding lemma that is needed when proving the discreteness of the set of transmission eigenvalues in Section 3, which then follows by analytic Fredholm theory. In Section 4 sufficient conditions for the existence of transmission eigenvalues are obtained, and the existence of an infinite set of transmision eigenvalues is studied in Section 5. 
In Section 6 we study the reformulation of the problem which is modeled on the reduced wave, rather than the Schr\"odinger, equation. It is interesting that this problem poses fewer restrictions on the operator than the Schr\"odinger version, and thus even in the second order case we do not need to assume the ellipticity. 
Section 7 is devoted to the study of two examples from physics, namely, the biharmonic operator and Dirac system. In the last two Sections 8 and 9 we explain the relation of transmission eigenvalues to scattering theory, relying upon a generalized Rellich theorem. 

\section{Interior transmission eigenvalues: The Schr\"odinger case}

\label{sec_Schrod}

As described in the introduction, let $V\in L^\infty(\R^n)$ be  compactly supported  with  $\supp(V)=\overline{\Omega}$, where $\Omega\subset \R^n$ is a bounded domain, and  let $P_0(D)$ be a partial differential operator  of order  $m\ge 2$ with constant real coefficients,
\[
P_0(D)=\sum_{|\alpha|\le m} a_{\alpha}D^\alpha, \quad a_\alpha\in\R, \quad D_j=-i\frac{\partial}{\partial x_j},\quad j=1,\dots,n.
\] 

The following estimate \cite[Theorem 2.1]{hor55}, \cite[Theorem 10.3.7]{horbookII}
\[
\|\varphi\|\le C_{\Omega}\|P_0(D)\varphi\|, \quad\varphi\in C^\infty_0(\Omega),
\]
where $C_\Omega$ is a constant depending on the domain $\Omega$,
implies that \eqref{eq_norm} is a norm on $C^\infty_0(\Omega)$, and as already  mentioned, completion of $C^\infty_0(\Omega)$ with respect to the norm \eqref{eq_norm} is denoted by $H^{P_0}_0(\Omega)$; see also \cite{jac88}.

Consider the interior transmission problem 
\begin{equation}
\label{eq_ITP}
\begin{aligned}
(P_0-\lambda)v=0 \quad &\text{in} \quad \Omega,\\
(P_0+V-\lambda)w=0 \quad &\text{in} \quad \Omega,\\
 v-w \in H^{P_0}_0(\Omega).
\end{aligned}
\end{equation}

We say that $\lambda\in \C$ is a {\em transmission eigenvalue} if the problem \eqref{eq_ITP} has non-trivial solutions $0\ne v\in L^2_{\textrm{loc}}(\Omega)$ and  $0\ne w\in L^2_{\textrm{loc}}(\Omega)$.  It suffices to require that $v\ne 0$. Indeed, assume that $w=0$. Then $v\in H^{P_0}_0(\Omega)$ and extending $v$ by zero to the complement of $\Omega$, we get $(P_0-\lambda)v=0$ in $\R^n$. Since every constant coefficient differential operator is injective $\mathcal{E}'(\R^n)\to \mathcal{E}'(\R^n)$, we have $v=0$. 

Throughout this work we shall assume that $V$ is real-valued and such that  $V\ge\delta>0$ a.e. in $\overline{\Omega}$.

We will use also the following equivalent characterization of the space \(H^{P_0}_0(\Omega)\).
Let 
\[
\tilde P_0(\xi)=(\sum_{|\alpha|\ge 0}|P_0^{(\alpha)}(\xi)|^2)^{1/2}, \quad P_0^{(\alpha)}(\xi)=D^{\alpha} P_0(\xi),\quad  \alpha\in \N^{n},
\]
and define on $C^\infty_0(\Omega)$ the norm
\begin{equation}
\label{eq_norm_2}
\|\varphi\|_{\tilde P_0}^2=\sum_{|\alpha|\ge 0}\|P_0^{(\alpha)}(D)\varphi\|^2.
\end{equation}

The estimate  \cite[Lemma 2.8]{hor55}
\begin{equation}
\label{eq_horm}
\|P_0^{(\alpha)}(D)\varphi\|\le C_{\Omega,\alpha}\|P_0(D)\varphi\|,\quad\varphi\in C^\infty_0(\Omega),
\end{equation}
yields that the norms \eqref{eq_norm} and \eqref{eq_norm_2} are equivalent on $H^{P_0}_0(\Omega)$

The space  $H^{P_0}_0(\Omega)$ is isometrically imbedded into the space $B_{2,\tilde P_0}(\R^n)$ via zero extension to the complement of $\Omega$.
Here 
\[
B_{2,\tilde P_0}(\R^n)=\{u\in S'(\R^n):\tilde P_0\hat u\in L^2(\R^n)\}
\]
is a Banach space with the norm
\[
u\mapsto\|\tilde P_0 \hat u\|_{L^2}
\]
introduced in \cite[Definition 10.1.1]{horbookII}. The proof of the discreteness of the set of transmission eigenvalues depends on the following result.

\begin{lem}
\label{prop_comp}
Assume that $\tilde P_0(\xi)\to\infty$ when $|\xi|\to \infty$. Then the imbedding 
\[
H^{P_0}_0(\Omega)\hookrightarrow L^2(\Omega)
\]
is compact. 
\end{lem}

\begin{proof}

The claim follows immediately from \cite[Theorem 10.1.10]{horbookII}. 

\end{proof}

A condition under which we can apply this result is found with the help of the set
\begin{equation}
\label{eq_Lambda_p_0}
\Lambda(P_0)=\{\xi\in\R^n:\lambda\mapsto P_0(\lambda\xi+\eta)\text{ is constant } \forall \eta\in \R^n\}
\end{equation}
Then \cite[Proposition 10.2.9]{horbookII} implies that $\tilde P_0(\xi)\to\infty$ as $|\xi|\to \infty$ provided that $\Lambda(P_0)=\{0\}$, i.e. there are no hidden variables in $P_0(\xi)$. 


\section{Discreteness of the set of transmission eigenvalues}

\begin{prop}
\label{thm_equivalence}
A complex number $\lambda$ is a transmission eigenvalue if and only if there exists $0\ne u\in H^{P_0}_0(\Omega)$ satisfying
\begin{equation}
\label{eq_teq}
(P_0+V-\lambda)\frac{1}{V}(P_0-\lambda)u=0\quad \text{in}\quad \mathcal{D}'(\Omega).
\end{equation}
\end{prop}

\begin{proof}
Assume that $\lambda$ is a transmission eigenvalue. Then there exists a non-trivial solution $(v,w)\ne 0$ of \eqref{eq_ITP}. Define $u=v-w\in H^{P_0}_0(\Omega)$. Then $(P_0-\lambda)u=Vw$ and thus, 
\[
(P_0+V-\lambda)\frac{1}{V}(P_0-\lambda)u=(P_0+V-\lambda)w=0.
\] 

To see the reverse implication, we first claim that
\[
(P_0+V-\lambda)\frac{1}{V}(P_0-\lambda)=(P_0-\lambda)\frac{1}{V}(P_0+V-\lambda).
\]
Indeed, 
\begin{align*}
&(P_0+V-\lambda)\frac{1}{V}(P_0-\lambda)u-(P_0-\lambda)\frac{1}{V}(P_0+V-\lambda)u\\
&=(P_0-\lambda)\frac{1}{V}(P_0-\lambda)u+\frac{V}{V}(P_0-\lambda)u-(P_0-\lambda)\frac{1}{V}(P_0-\lambda)u\\
&-(P_0-\lambda)\frac{V}{V}u=0.
\end{align*}
Let $0\ne u\in H^{P_0}_0(\Omega)$ be a solution of \eqref{eq_teq}. 
Denote now 
\[
w=\frac{1}{V}(P_0-\lambda)u\quad \text{and}\quad v=\frac{1}{V}(P_0+V-\lambda)u.
\]
Then 
\[
(P_0-\lambda)v=0, \quad (P_0+V-\lambda)w=0\quad \text{in}\quad \Omega
\]
and
\[
v-w=\frac{1}{V}(P_0+V-\lambda-P_0+\lambda)u=u\in H^{P_0}_0(\Omega).
\]

\end{proof}

Set
\[
T_\lambda=(P_0+V-\lambda)\frac{1}{V}(P_0-\lambda).
\]
Let $\lambda\in \C$. Then for $\varphi,\psi\in C^\infty_0(\Omega)$, we define a sesquilinear form
\[
B_\lambda(\varphi,\psi)=\langle T_\lambda\varphi,\psi \rangle_{L^2(\Omega)}=
\langle\frac{1}{V}(P_0-\lambda)\varphi ,(P_0+V-\lambda)\psi \rangle_{L^2(\Omega)}
\]
which 
extends uniquely to a continuous sesquilinear form on $H^{P_0}_0(\Omega)\times H^{P_0}_0(\Omega)$. 

\begin{lem}
\label{lem_coercive}
Assume that either 
$P_0(\xi)$ or \(-P_0(\xi)\) is bounded from below on $\R^n$.
Then there exists $\lambda_0\in \R$ such that the sesquilinear form $B_{\lambda_0}(u,v)$ is coercive on $H^{P_0}_0(\Omega)\times H^{P_0}_0(\Omega)$.
\end{lem}

\begin{proof}
Let $\varphi\in C^\infty_0(\Omega)$ and set $C_0=\inf_\Omega1/V>0$. Then
\begin{equation}
\label{eq_est_b}
\begin{aligned}
B_\lambda(\varphi,\varphi)&=\langle(P_0-\lambda)\frac{1}{V}(P_0-\lambda)\varphi ,\varphi \rangle_{L^2(\Omega)}+\langle (P_0-\lambda)\varphi,\varphi \rangle_{L^2(\Omega)}\\
&\ge C_0\langle(P_0-\lambda)^2\varphi ,\varphi \rangle_{L^2(\Omega)}+\langle (P_0-\lambda)\varphi,\varphi \rangle_{L^2(\Omega)}\\
&= C_0\|P_0\varphi\|^2-2C_0\lambda \langle P_0\varphi,\varphi \rangle_{L^2(\Omega)} \\ &\quad\quad\quad+  \lambda^2C_0\|\varphi\|^2-\lambda\|\varphi\|^2+ \langle P_0\varphi,\varphi \rangle_{L^2(\Omega)}.
\end{aligned}
\end{equation}
To estimate the last term in  the right hand side of \eqref{eq_est_b} we write
\begin{equation}
\label{eq_est_b_4}
\langle P_0\varphi,\varphi \rangle_{L^2(\Omega)}\le \|P_0\varphi\|\|\varphi\|\le \frac{C_0}{2}\|P_0\varphi\|^2+\frac{2}{C_0}\|\varphi\|^2.
\end{equation}
For the second  term in the right hand side of \eqref{eq_est_b}, assuming, to fix the ideas, that $P_0(\xi)$ is bounded from below on $\R^n$, we get
\begin{equation}
\label{eq_est_b_2}
\langle P_0\varphi,\varphi \rangle_{L^2(\Omega)}=(2\pi)^{-n}\int_{\R^n} P_0(\xi)|\hat\varphi(\xi)|^2d\xi\ge -C\|\varphi\|^2.
\end{equation}
Now combining \eqref{eq_est_b}, \eqref{eq_est_b_4} and \eqref{eq_est_b_2}, we can find $\lambda_0\in \R$ with $|\lambda_0|$ large enough such that
\[
B_{\lambda_0}(\varphi,\varphi)\ge \frac{C_0}{2}\|P_0\varphi\|^2=\frac{C_0}{2}\|\varphi\|^2_{P_0}, \quad \forall \varphi\in C^\infty_0(\Omega).
\] 
This establishes the coercivity of $B_{\lambda_0}$ on $H^{P_0}_0(\Omega)\times H^{P_0}_0(\Omega)$.  

\end{proof}

Let us denote by $H^{-P_0}(\Omega)$ the dual of the space $H_0^{P_0}(\Omega)$. Notice that $H^{-P_0}(\Omega)$ can be viewed as a subspace of the space of distributions $\mathcal{D}'(\Omega)$.

\begin{thm} 
\label{thm_discret_1}
Assume that $\tilde P_0(\xi)\to\infty$ when $|\xi|\to \infty$ and that
the symbol 
$\pm P_0(\xi)$ is bounded from below on $\R^n$, for one of the choices of the sign. Then the set of transmission eigenvalues is discrete.

\end{thm}

\begin{proof}
The proof is based on application of the analytic Fredholm  theory to the holomorphic family of operators
\[
T_\lambda:\C\to \mathcal{L}(H^{P_0}_0(\Omega),H^{-P_0}(\Omega)), \quad
T_\lambda=(P_0+V-\lambda)\frac{1}{V}(P_0-\lambda).
\]
Since by Lemma \ref{lem_coercive} there exists $\lambda_0\in \R$ such that the sesquilinear form $B_{\lambda_0}$, corresponding to the operator $T_{\lambda_0}$, is coercive and bounded on $H^{P_0}_0(\Omega)\times H^{P_0}_0(\Omega)$, by an application of the Lax--Milgram lemma we conclude that $T_{\lambda_0}:H^{P_0}_0(\Omega)\to H^{-P_0}(\Omega)$ is invertible.

Now, for any $\lambda$ the operator 
\begin{equation}
\label{eq_perturbation}
T_\lambda-T_{\lambda_0}=-P_0\frac{\lambda-\lambda_0}{V}-\frac{\lambda-\lambda_0}{V}P_0-(\lambda-\lambda_0)+\frac{\lambda^2-\lambda_0^2}{V}:H^{P_0}_0(\Omega)\to H^{-P_0}(\Omega)
\end{equation}
is compact. Indeed, 
the last two terms in the right hand side of  \eqref{eq_perturbation} are compact because
\[
 -(\lambda-\lambda_0)+\frac{\lambda^2-\lambda_0^2}{V}:L^2(\Omega)\to L^2(\Omega) 
 \]
 is bounded and by Lemma \ref{prop_comp} the embedding $H^{P_0}_0(\Omega)\hookrightarrow  L^2(\Omega)$ is compact. 
 The second term  in the right hand side of \eqref{eq_perturbation} defines the operator $-\frac{\lambda-\lambda_0}{V}P_0$. This operator is compact as a composition of the bounded operators
\[
P_0:H^{P_0}_0(\Omega)\to L^2(\Omega),\quad
-\frac{\lambda-\lambda_0}{V}:L^2(\Omega)\to L^2(\Omega)
\]
with the compact inclusion $L^2(\Omega)$ into $H^{-P_0}(\Omega)$. 
The first term in  the right hand side of \eqref{eq_perturbation} defines the operator $-P_0\frac{\lambda-\lambda_0}{V}$. This is a bounded operator
\[
-P_0\frac{\lambda-\lambda_0}{V}:L^2(\Omega)\to H^{-P_0}(\Omega).
\]
Since the inclusion $H^{P_0}_0(\Omega)$ to $L^2(\Omega)$ is compact, we conclude the compactness of $-P_0\frac{\lambda-\lambda_0}{V}$
as an operator from $H^{P_0}_0(\Omega)$ to $H^{-P_0}(\Omega)$. Thus also \(T_\lambda - T_{\lambda _0}\) is compact.

Hence, the operator $T_{\lambda}$ is Fredholm as the sum of an invertible operator and a compact operator. Since it is invertible at $\lambda=\lambda_0$, the analytic Fredholm theory guarantees that $T_\lambda$ is invertible except for  a discrete set of values $\lambda$.  

\end{proof}

\begin{rem}
The assumptions of Theorem \ref{thm_discret_1} are satisfied automatically when the operator $P_0$ is elliptic. 
\end{rem}

\begin{rem}
The condition that  $\pm P_0(\xi)$ is bounded from below on $\R^n$, for one of the choices of the sign, cannot be removed completely. Indeed,  let \(P_0 = D_{x_1} + \Delta _{x'}\), $x=(x_1,x')\in \R^n$. 
 Then if \((v,w)\) is a pair of transimssion eigenfunctions for a transmission eigenvalue \(\lambda\) on some domain $\Omega$, then  \(v'(x) = e^{i\mu x_1} v(x)\) and \(w'(x) = e^{i\mu x_1} w(x)\) will be transmission eigenfunctions for the transmission eigenvalue \(\lambda + \mu\). Hence the set of transmission eigenvalues is either empty or \(\C\).
\end{rem}

\begin{rem} Notice that $\text{ind}(T_\lambda)=0$, since $T_\lambda$ is the sum of an invertible operator and a compact operator. It follows that $T_\lambda$ fails to be invertible if and only if $T_\lambda$ is not injective. 
\end{rem}

\section{Existence of transmission eigenvalues }

\label{sec_exis}

In this section we study the question of existence of transmission eigenvalues. Let  us write the operator $T_\lambda$ in the following form,
\begin{equation}
\label{eq_T_lambda}
T_\lambda=(P_0+\frac{V}{2}-\lambda)\frac{1}{V}(P_0+\frac{V}{2}-\lambda)-\frac{V}{4}.
\end{equation}

\begin{lem}
\label{lem_T(tau)_self-adjoint}
The operator $T_\lambda$ given by \eqref{eq_T_lambda}, equipped 
with the domain 
\[
\mathcal{D}(T_\lambda)=\{u\in H^{P_0}_0(\Omega): P_0(\frac{1}{V}P_0-\frac{\lambda}{V})u\in L^2(\Omega)\},
\]  
is an unbounded self-adjoint operator on $L^2(\Omega)$, for any $\lambda\in \R$.
\end{lem}

\begin{proof} Let $A_\lambda=\frac{1}{V^{1/2}}(P_0+\frac{V}{2}-\lambda)$. Then $A_\lambda$ with the domain $\mathcal{D}(A_\lambda)=H_0^{P_0}(\Omega)$ is a densely defined closed operator. Its $L^2$--adjoint is given by $A_\lambda^*=(P_0+\frac{V}{2}-\lambda)\frac{1}{V^{1/2}}$ with the domain 
$\mathcal{D}(A_\lambda^*)=\{u\in L^2(\Omega):P_0\frac{1}{V^{1/2}}u\in L^2(\Omega)\}$. Then the result follows from the fact that $T_\lambda=A_\lambda^*A_\lambda-\frac{V}{4}$.
\end{proof}

For future reference, let us remark that the form domain of the semibounded self-adjoint operator $T_\lambda$ is $H^{P_0}_0(\Omega)$, for any $\lambda\in \R$.
Let us also recall the corresponding quadratic form
\[
B_\lambda(u,u)=
\langle\frac{1}{V}(P_0-\lambda)u ,(P_0+V-\lambda)u \rangle_{L^2(\Omega)}, \quad u\in H^{P_0}_0(\Omega). 
\]
From Lemma \ref{lem_coercive}, we know that there exists $\lambda_0\in \R$ such that $B_{\lambda_0}$ is coercive on $H^{P_0}_0(\Omega)$.

Similarly to \cite{paisyl08}, we have the following result. 

\begin{thm}
\label{thm_existence1}
 If for some $\lambda\in \R$, there is $u\in H^{P_0}_0(\Omega)$ such that 
\begin{equation}
\label{eq_existence2}
 B_\lambda(u,u) \le 0
\end{equation}
then there exists a transmission eigenvalue $\lambda^*\in (\lambda_0, \lambda]$, if $\lambda_0<\lambda$, or $\lambda^*\in [\lambda, \lambda_0)$, if $\lambda_0>\lambda$.

\end{thm}

\begin{proof}

The operator $T_\lambda$ has a compact resolvent acting on $L^2(\Omega)$ because the canonical embedding of its domain $\mathcal{D}(T_\lambda)$ to $L^2(\Omega)$ is compact. The facts that $T_\lambda$ depends continuously on $\lambda$, $T_\lambda$ is self-adjoint with compact resolvent and the form domain of $T_\lambda$ does not depend on $\lambda$  imply, through an application of the variational principle,  
 that the eigenvalues of $T_\lambda$ depend continuously on $\lambda$  for $\lambda\in \R$.  See \cite{paisyl08} for the details of this argument. 

The hypothesis \eqref{eq_existence2} yields that  $T_\lambda$ has at least one non-positive eigenvalue. Since $T_{\lambda_0}$ is positive definite, the lowest eigenvalue, which is a continuous function of $\lambda \in \R$, must pass through zero for some $\lambda^*\in (\lambda_0, \lambda]$, if $\lambda_0<\lambda$, or $\lambda^*\in [\lambda, \lambda_0)$, if $\lambda_0>\lambda$.  Hence, the operator $T_\lambda^*$ is non-injective and therefore, $\lambda^*$ is a transmission eigenvalue. 

\end{proof}

We define the multiplicity of a transmission eigenvalue $\lambda$ to be the multiplicity of zero as an eigenvalue of $T_\lambda$. 
Since the self-adjoint operator $T_\lambda$, $\lambda\in \R$, has a compact resolvent, the multiplicity of $\lambda$ is finite. 

As in \cite[Lemma 14]{paisyl08}, we get the following result on  existence of more than one transmission eigenvalues. 

\begin{lem}
\label{lem_existence_p}
If there exists $\lambda\in \R$ and a p-dimensional subspace $V^p\subset H^{P_0}_0(\Omega)$ such that 
\begin{equation}
\label{eq_existence3}
B_\lambda (u,u) \le 0
\end{equation}
for all $u\in V^p$, 
then there exist  $p$ transmission eigenvalues, counting with multiplicity.

\end{lem}

The next result tells us that if the potential is strong enough then the transmission eigenvalues exist.

\begin{lem}

Let $V(x)\ge \delta>0$. If there exists a $p$-dimensional subspace $V^p\subset H^{P_0}_0(\Omega)$ such that 
\[
\delta\ge 2\frac{\|P_0u\|}{\|u\|}
\]
for all $0\ne u\in V^p$, then there exist  $p$ transmission eigenvalues, counting with multiplicity. 

\end{lem}

\begin{proof} We have 
\begin{align*}
B_\lambda(u,u)&=\langle \frac{1}{V}(P_0-\lambda)u,(P_0-\lambda)u \rangle + \langle (P_0-\lambda)u,u \rangle\\
&\le \frac{1}{\delta}\|P_0u\|^2-\frac{2\lambda}{\delta}\langle P_0 u,u \rangle+\frac{\lambda^2}{\delta}\|u\|^2+\langle P_0 u,u \rangle-\lambda\|u\|^2. 
\end{align*}
Choosing $ \lambda=\delta/2 $, we get
\[
B_\lambda(u,u)
\le \frac{1}{\delta}\|P_0u\|^2-\frac{\delta}{4}\|u\|^2.
\]
 The claim follows by an application of Lemma \ref{lem_existence_p}. 
\end{proof}

\section{Existence of infinitely many transmission eigenvalues}

\label{sec_exis_inf}

Now assume, as before, that $P_0(D)$ is a partial differential operator of degree $m\ge 2$ with constant real coefficients and $\tilde P_0(\xi)\to\infty$ when $|\xi|\to \infty$ and 
the symbol 
$\pm P_0(\xi)$ is bounded from below on $\R^n$, for one of the choices of the sign.
Let $\varepsilon >0$ and consider an open ball $B_\varepsilon(0)\subset \R^n$ of radius $\varepsilon$ centered at the origin. Consider the following interior transmission problem for the ball $B_\varepsilon(0)$ and an arbitrary constant potential $\delta>0$,
\begin{equation}
\label{eq_ITP_ball}
\begin{aligned}
(P_0-\lambda)v=0 \quad &\text{in} \quad B_\varepsilon(0),\\
(P_0+\delta-\lambda)w=0 \quad &\text{in} \quad B_\varepsilon(0),\\
 v-w \in H^{P_0}_0(B_\varepsilon(0)).
\end{aligned}
\end{equation}
Notice that as the operator $P_0$ has constant coefficients and $\delta$ is a constant, if there exists a transmission eigenvalue $\lambda(\varepsilon)$ for the ball $B_\varepsilon(0)$, then $\lambda(\varepsilon)$ is a transmission eigenvalue for an arbitrary ball in $\R^n$ of radius $\varepsilon$.

Existence of an infinite discrete set of real transmission eigenvalues when $P_0=-\Delta$  in $\R^n$,  $n=2,3$,  was obtained in \cite{CakDroHou}.  The proof relies on an explicit computation in the case of a constant potential when $\Omega$ is a ball, combined with  a specific analytic framework of generalized eigenvalue problems. 

Following the approach of \cite{CakDroHou}, we have the following conditional result on existence of an infinite set of real transmission eigenvalues for general operators.  In the proof we give a direct argument relying upon  Lemma \ref{lem_existence_p}.

 \begin{prop}
\label{prop_infty}
Let  $V\in L^\infty(\R^n)$ be real-valued compactly supported with $\supp(V)=\overline{\Omega}$ and assume that 
\[
V\ge \delta>0\quad \text{a.e. in} \ \overline{\Omega}.
\]
Furthermore, assume that for any $\varepsilon >0$, there exists a real transmission eigenvalue for  \eqref{eq_ITP_ball}. Then the problem \eqref{eq_ITP} has 
 an infinite set of real transmission eigenvalues.  
\end{prop}

\begin{proof}

For every $p\in \N$, there exists $\varepsilon >0$ small enough such that $\Omega$ contains $p$ disjoint balls $B^1_\varepsilon,\dots,B^p_\varepsilon$ of radius $\varepsilon$, i.e. $\overline{B^i_\varepsilon}\subset \Omega$, $i=1,\dots,p$, and $\overline{B^i_\varepsilon}\cap \overline{B^j_\varepsilon}=\emptyset$ for $i\ne j$.
Let $\lambda(\varepsilon)\in \R$ be a transmission eigenvalue for \eqref{eq_ITP_ball}. Then it  is a transmission eigenvalue for each of the balls $B^i_\varepsilon$ with potential $\delta$. Thus, Proposition \ref{thm_equivalence} implies that there are $0\ne u^i(\varepsilon)\in H^{P_0}_0(B^i_\varepsilon)$ such that
\[
(P_0+\delta-\lambda(\epsilon))\frac{1}{\delta}(P_0-\lambda)u^i(\varepsilon)=0 \quad \text{in}\quad  \mathcal{D}'(B^i_{\varepsilon}), \quad i=1,\dots,p.
\]
The extension by zero $\tilde u^i$ of $u^i(\varepsilon)$ to the whole of  $\Omega$ is in $H^{P_0}_0(\Omega)$. Moreover, the functions $\tilde u^1,\dots, \tilde u^p$ 
form an  orthogonal system in $H^{P_0}_0(\Omega)$, since they have disjoint supports. This implies that
\[
\langle \frac{1}{\delta}(P_0-\lambda(\varepsilon))\tilde u^i,(P_0+\delta-\lambda(\varepsilon))\tilde u^j \rangle_{L^2(\Omega)}=0, \quad i,j=1,\dots,p.
\]
Set $V^p=\textrm{span}\{\tilde u^1,\dots,\tilde u^p\}$.  Let $u\in V^p$. Then
\begin{align*}
B_{\lambda(\varepsilon)}(u,u)&=\langle\frac{1}{V}(P_0-\lambda(\varepsilon))u,(P_0+V-\lambda(\varepsilon))u\rangle_{L^2(\Omega)}\\
&=\langle\frac{1}{V}(P_0-\lambda(\varepsilon))u,(P_0-\lambda(\varepsilon))u\rangle_{L^2(\Omega)}+\langle\frac{1}{V}(P_0-\lambda(\varepsilon))u,V u\rangle_{L^2(\Omega)}\\
&\le\langle\frac{1}{\delta}(P_0-\lambda(\varepsilon))u,(P_0-\lambda(\varepsilon))u\rangle_{L^2(\Omega)}+\langle(P_0-\lambda(\varepsilon))u,u\rangle_{L^2(\Omega)}\\
&=\langle\frac{1}{\delta}(P_0-\lambda(\varepsilon))u,(P_0+\delta-\lambda(\varepsilon))u\rangle_{L^2(\Omega)}=0.
\end{align*}
Hence, it follows from Lemma \ref{lem_existence_p} that problem \eqref{eq_ITP}  has $p$ transmission eigenvalues, counting with multiplicity. 
As $p$ is arbitrary, the result follows.

\end{proof}

\section{The generalized acoustic problem}
Let $V\in L^\infty(\R^n)$ be  compactly supported in $\R^n$ with  $\supp(V)=\overline{\Omega}$, where $\Omega\subset \R^n$ is a bounded domain, and  $P_0(D)$ be a partial differential operator  of degree $m\ge 2$ with constant real coefficients. 

As in physics, while the problem 
\[
(-\Delta-\lambda +V)u=0,
\]
models the time-independent Schr\"odinger equation, the equation
\[
(-\Delta-\lambda(1+V))u=0
\]
describes acoustic wave propagation with refractive index $1+V$. In this section, we study the interior transmission problem for the latter, where $-\Delta$
is replaced by a general $P_0$. This problem has the following form, 
\begin{equation}
\label{eq_ITP_Helm}
\begin{aligned}
(P_0-\lambda)v=0 \quad &\text{in} \quad \Omega,\\
(P_0-\lambda(1+ V))w=0 \quad &\text{in} \quad \Omega,\\
 v-w \in H^{P_0}_0(\Omega).
\end{aligned}
\end{equation}

We say that $\lambda\in \C$ is a transmission eigenvalue if the problem \eqref{eq_ITP_Helm} has non-trivial solutions $0\ne v\in L^2_{\textrm{loc}}(\Omega)$ and  $0\ne w\in L^2_{\textrm{loc}}(\Omega)$.  It suffices to require that $v\ne 0$. 
Notice that $\lambda=0$ is always a transmission eigenvalue for \eqref{eq_ITP_Helm}. 

We shall assume that $V$ is real-valued and such that  $V\ge\delta>0$ a.e. in $\overline{\Omega}$.

As in Proposition \ref{thm_equivalence}, one can show that  $0\ne \lambda\in \C$ is a transmission eigenvalue if and only if there exists $0\ne u\in H^{P_0}_0(\Omega)$ satisfying
\[
T_\lambda u=(P_0-\lambda(1+V))\frac{1}{V}(P_0-\lambda)u=0\quad \text{in}\quad \mathcal{D}'(\Omega).
\]

Define a sesquilinear form
\[
B_\lambda(\varphi,\psi)=\langle T_\lambda \varphi,\psi \rangle_{L^2(\Omega)}=
\langle \frac{1}{V}(P_0-\lambda)\varphi,(P_0-\lambda(1+V))\psi\rangle_{L^2(\Omega)}, \quad \varphi,\psi\in C^\infty_0(\Omega),
\] 
which extends uniquely to a continuous sesquilinear for on $H^{P_0}_0(\Omega)\times H^{P_0}_0(\Omega)$. 
Then $B_0$ is coercive in the sense that there exists $C_0>0$ such that
\[
B_0(\varphi,\varphi)\ge C_0\|\varphi\|^2_{P_0}, \quad \varphi\in H^{P_0}_0(\Omega). 
\]
Arguing as in Theorem \ref{thm_discret_1}, we get the following result.

\begin{thm} 

\label{thm_disc_Helm}
Assume that $\tilde P_0(\xi)\to\infty$ when $|\xi|\to \infty$. Then the set of transmission eigenvalues is discrete. 
\end{thm}
Remark that the multiplicity of each transmission eigenvalue except $\lambda=0$ is finite.

Notice that the assumption in Theorem \ref{thm_discret_1} that the symbol $\pmÊP_0(\xi)$ is bounded from bellow on $\R^n$, for one of the choices of the sign, is not needed in the Helmholtz case. In particular, the results of this section are applicable when $P_0=D_{x_1}^2  + \Delta _{x'}$ is the wave operator on $\R^n$. 

An inspection of the arguments of Sections \ref{sec_exis} -- \ref{sec_exis_inf} shows that the results there remain valid in the Helmholtz case. In particular,
Theorem \ref{thm_existence1} and Lemma \ref{lem_existence_p} hold true in the Helmholtz case as they stand. 
To formulate an analog of Proposition \ref{prop_infty},  consider the following interior transmission problem for an open ball $B_\varepsilon(0)\subset \R^n$ of radius $\varepsilon>0$ centered at the origin, and an arbitrary constant potential $\delta>0$,
\begin{equation}
\label{eq_ITP_Helm_ball}
\begin{aligned}
(P_0-\lambda)v=0 \quad &\text{in} \quad B_\varepsilon(0),\\
(P_0-\lambda(1+\delta))w=0 \quad &\text{in} \quad B_\varepsilon(0),\\
 v-w \in H^{P_0}_0(B_\varepsilon(0)).
\end{aligned}
\end{equation}

\begin{prop}
\label{prop_helm_inf_n}
 Assume that for any $\varepsilon >0$ and  $\delta>0 $, there exists a non-zero real transmission eigenvalue for  \eqref{eq_ITP_Helm_ball}. Then there exists an infinite set of real transmission eigenvalues for \eqref{eq_ITP_Helm}. 
\end{prop}

\section{Examples}

\subsection{The biharmonic operator}

As an example of a higher order operator to which our conditional existence results apply, we shall consider the biharmonic operator in $\R^3$, which arises, e.g., in the study of thin elastic plates.

\begin{prop}
Let $P_0=\Delta^2$ on $\R^3$. Then the problem \eqref{eq_ITP_Helm} has 
 an infinite set of real transmission eigenvalues.  

\end{prop}

 \begin{proof}

When proving this result,  we shall apply  Proposition \ref{prop_helm_inf_n}.  It therefore suffices to prove that the problem  \eqref{eq_ITP_Helm_ball} has a non-trivial solution for arbitrary $\varepsilon>0$ and $\delta>0$.  Without loss of generality, we may assume that $\delta$ is sufficiently small.  Let
$k=\lambda^{1/4}$, $k_\delta=(\lambda(1+\delta))^{1/4}$, and
\begin{equation}
\label{eq_rho}
\rho=\frac{k_\delta}{k}=(1+\delta)^{1/4}=1+\frac{\delta}{4}+\mathcal{O}(\delta^2), \quad 0<\delta\ll 1.
\end{equation}
Then writing
\[
\Delta^2 - k^4 = (\Delta - k^2) (\Delta + k^2) =  (\Delta + k^2) (\Delta - k^2),
\]
and for \(\Delta^2 - k_\delta^4\) similarly,  one sees that a reasonable ansatz for \(v\) and \(w\) in \eqref{eq_ITP_Helm_ball} is 
\[
v(x) = a_1j_0 ( kr) + a_2j_0 (ikr), \ w(x) = a_3 j_0 ( k_\delta r) + a_4j_0 (ik_\delta r),
\]
where \(j_0\) is the spherical Bessel function of order zero, and \(a_j\) are constants.  Our boundary conditions now take form 
\begin{equation}\label{biharmoninen_vakioreuna}
\left(\frac{d}{dr}\right) ^l (a_1j_0 ( kr) + a_2j_0 (ikr) - a_3 j_0 ( k_\delta r) - a_4j_0 (ik_\delta r)) |_{r=\varepsilon}= 0, \quad l=0,\,1,\, 2,\, 3.
\end{equation}
Using the known asymptotics for \(j_0 ( k\varepsilon)\) and  \(j_0 ( ik\varepsilon)\) and their derivatives, as $k\to\infty$, one sees that the deteminant of the linear system \eqref{biharmoninen_vakioreuna} is 
\begin{align*}
d=&\begin{vmatrix}
\frac{\sin (k\varepsilon)}{k\varepsilon} & \frac{e^{k\varepsilon}}{2k\varepsilon}+ \mathcal{O}(e^{-k\varepsilon}) & -\frac{\sin (k_\delta\varepsilon)}{k_\delta\varepsilon}
& -\frac{e^{k_{\delta}\varepsilon}}{2k_\delta\varepsilon}+  \mathcal{O}(e^{-k_\delta\varepsilon}) \\
\frac{\cos( k\varepsilon)}{\varepsilon} +\mathcal{O}(\frac{1}{k}) & \frac{e^{k\varepsilon}}{2\varepsilon}+  \mathcal{O}(\frac{e^{k\varepsilon}}{k})
& -\frac{\cos( k_\delta\varepsilon)}{\varepsilon} +\mathcal{O}(\frac{1}{k})
& -\frac{e^{k_{\delta}\varepsilon}}{2\varepsilon}+  \mathcal{O}(\frac{e^{k_\delta\varepsilon}}{k_\delta}) \\
-\frac{k\sin (k\varepsilon)}{\varepsilon} + \mathcal{O}(1)  & \frac{ke^{k\varepsilon}}{2\varepsilon} +  \mathcal{O}(e^{k\varepsilon}) & \frac{k_\delta\sin (k_\delta\varepsilon)}{\varepsilon} +\mathcal{O}(1)
& -\frac{k_\delta e^{k_{\delta}\varepsilon}}{2\varepsilon}+  \mathcal{O}(e^{k_\delta\varepsilon})\\
-\frac{k^2\cos( k\varepsilon)}{\varepsilon} + 
\mathcal{O}(k)
 & \frac{k^2e^{k\varepsilon}}{2\varepsilon} + \mathcal{O}(ke^{k\varepsilon})
& \frac{k_\delta^2\cos( k_\delta\varepsilon)}{\varepsilon} +\mathcal{O}(k)
& 
 -\frac{k_\delta^2 e^{k_{\delta}\varepsilon}}{2\varepsilon} + \mathcal{O}(k_\delta e^{k_\delta\varepsilon})
\end{vmatrix}\\
=
&\frac{1}{4\varepsilon^4}e^{(k+k_\delta)\varepsilon}k^2\frac{\rho-1}{\rho} \left(d_1+\mathcal{O}(\frac{1}{\delta k})\right),
\end{align*}
where $d_1$ is given by
\small
\[
\begin{vmatrix} (\rho-1)\sin(k\varepsilon) -(1+\rho)\cos(k\varepsilon) & (\rho-1)\sin(k_\delta \varepsilon) +(\rho+1)\cos (k_\delta \varepsilon)\\
-(\rho^2+\rho+2)\cos (k\varepsilon) +(\rho^2+\rho)\sin(k\varepsilon) & (2\rho^2+\rho+1)\cos(k_\delta\varepsilon) -(\rho+1)\sin(k_\delta\varepsilon)
\end{vmatrix}
\]
\normalsize
Using \eqref{eq_rho}, we get
\[
d_1=4\sin((\rho-1)k\varepsilon)+\mathcal{O}(\delta). 
\]
Hence, to show that the linear system \eqref{biharmoninen_vakioreuna} has a non-trivial solution, it suffices to check that  the function 
\[
\sin((\rho-1)k\varepsilon)+\mathcal{O}(\delta)+ \mathcal{O}(\frac{1}{\delta k})
\]
has real zeros, for $k$ large enough. The latter is clear, however, from the periodicity of the function 
\[
k\mapsto \sin((\rho-1)k\varepsilon). 
\]
We may also notice that the minimal period of this function is 
\[
\frac{8\pi}{\delta \varepsilon}+\mathcal{O}(\frac{1}{\varepsilon}).
\] 
This completes the proof.

\end{proof}

\subsection{The Dirac system}
Our approach generalizes also to many systems. We demonstrate this here by carrying out the analysis in the case of the Dirac system.   

The free Dirac operator in $\R^3$ is given by the $4\times 4$ matrix
\[
\mathcal{L}_0(D)=\begin{pmatrix} 0 & \sigma\cdot D\\
\sigma\cdot D & 0
\end{pmatrix},
\]
where $D=-i\nabla$ and $\sigma=(\sigma_1,\sigma_2,\sigma_3)$ is a vector of Pauli matrices with 
\[
\sigma_1=\begin{pmatrix}
0& 1\\
1& 0
\end{pmatrix},\quad
\sigma_2=\begin{pmatrix}
0& -i\\
i& 0
\end{pmatrix},\quad
\sigma_3=\begin{pmatrix}
1& 0\\
0& -1
\end{pmatrix}.
\]
The most important property of the Dirac operator is the following one,
\[
\mathcal{L}_0(D)^2=-\Delta I_4,
\]
where $I_4$ is the $4\times 4$ identity matrix.

Let $\Omega\subset \R^3$ be a bounded domain in $\R^3$ with a connected $C^\infty$-smooth boundary. 
It is known \cite{NakTsu00} that when equipped with the domain 
\[
\mathcal{D}(\mathcal{L}_0)=\{\begin{pmatrix} u_+\\
u_-
\end{pmatrix}\in L^2(\Omega)^2\times L^2(\Omega)^2: \ u_+\in H_0^1(\Omega)^2, \sigma\cdot Du_-\in L^2(\Omega)^2
\},
\]
the Dirac operator $\mathcal{L}_0$ is self-adjoint on $L^2(\Omega)^4$.

Notice that 
\[
H_0^1(\Omega)^2\times H^1(\Omega)^2\subset \mathcal{D}(\mathcal{L}_0).
\]
However, in general,  $\mathcal{D}(\mathcal{L}_0)$ is strictly larger than the Sobolev space $H_0^1(\Omega)^2\times H^1(\Omega)^2$, see \cite{Sch95} for the discussion and a precise example.

Let $V(x)$ be an Hermitian $4\times4$-matrix-valued function whose entries belong to $L^\infty(\R^3)$. 
An application of the Kato-Rellich theorem shows that the operator $\mathcal{L}_0(D)+V$ is self-adjoint on $\mathcal{D}(\mathcal{L}_0)$.   
Assume that $0$ is not in the spectrum of 
$\mathcal{L}_0(D)+V$. 
Then
it was shown in \cite{NakTsu00, salotzou} that for any $f\in H^{1/2}(\p \Omega)^2$,
 the boundary value problem
\begin{equation}
\label{eq_bvp_D}
\begin{aligned}
&(\mathcal{L}_0(D)+V)u=0,\quad \text{in}\quad \Omega,\\
&u_+=f,\quad \text{on}\quad \p \Omega,
\end{aligned}
\end{equation}
has
a unique solution $u\in H^1(\Omega)^4$. 
The set of the Cauchy data for \eqref{eq_bvp_D} is given by
\[
\{(u_+|_{\p \Omega},u_-|_{\p \Omega}):u\in  H^1(\Omega)^4 \ \text{is a solution of  } (\mathcal{L}_0(D)+V)u=0 \text{ in }\Omega\}.
\]

Assume now  that $\supp(V)=\overline{\Omega}$. 
The interior transmission problem for the Dirac operator is the following boundary value problem,
\begin{equation}
\label{eq_ITP_Helm_D}
\begin{aligned}
(\mathcal{L}_0(D)-\lambda I_4)v=0 \quad \text{in}\quad \Omega,\\
(\mathcal{L}_0(D)-\lambda(I_4+V))w=0 \quad \text{in}\quad \Omega, \\
v-w\in H^1_0(\Omega)^4.
\end{aligned}
\end{equation}
We say that $\lambda\in \C$ is a transmission eigenvalue if the problem \eqref{eq_ITP_Helm_D} has non-trivial solutions $0\ne v\in L^2_{\textrm{loc}}(\Omega)^4$
and  $0\ne w\in L^2_{\textrm{loc}}(\Omega)^4$. It suffices to require that $v\ne 0$.

Notice that $\lambda=0$ is a transmission eigenvalue and the space of functions
\[
\{v\ne 0:\mathcal{L}_0(D)v=0\}
\]
is infinite dimensional.

\begin{rem}
The standard electromagnetic potential for the Dirac operator given by
\[
V=\mathcal{L}_0(A)+Q,
\]
with 
\[
Q=\begin{pmatrix} q_+I_2 & 0\\
0 & q_-I_2
\end{pmatrix},\quad A=(a_1,a_2,a_3)\in L^\infty(\Omega;\R^3),\quad
q_\pm \in L^\infty(\Omega;\R)
\]
 is included in the setup above. 
 
 \end{rem}

Throughout this section,  we shall assume that $V(x)$ is an Hermitian positive-definite $4\times 4$-matrix valued function, i.e. 
there is a constant $c_V>0$ such that
\[
\langle V(x)\eta,\eta \rangle\ge c_V|\eta|^2,\quad \forall x\in \overline{\Omega},\quad \forall \eta\in \C^4,
\] 
where $\langle \cdot,\cdot  \rangle$ is the inner product in $\C^4$. Moreover, we shall assume that the entries of $V(x)$ belong to $C^\infty(\overline\Omega)$. Thus, the entries of the inverse matrix $V^{-1}(x)$ also belong to $C^\infty(\overline\Omega)$.

Arguing as in the earlier sections, we see that the following characterization of transmission eigenvalues holds: 
 $0\ne \lambda\in \C$ is a transmission eigenvalue if and only if there exists $0\ne u\in H^1_0(\Omega)^4$ satisfying 
\[
T_\lambda u=(\mathcal{L}_0-\lambda(I_4+V))V^{-1}(\mathcal{L}_0-\lambda I_4)u=0 \quad \text{in}\quad \mathcal{D}'(\Omega)^4. 
\]
Here 
\[
T_\lambda =A-\lambda B+\lambda^2C,
\]
where
\begin{align*}
A&=\mathcal{L}_0V^{-1}\mathcal{L}_0,\\
B&=V^{-1}\mathcal{L}_0+\mathcal{L}_0V^{-1}+\mathcal{L}_0,\\
C&=1+V^{-1}.
\end{align*}

\begin{prop}
  The operator $A$ is symmetric and positive on $L^2(\Omega)^4$, when equipped with the domain $C^\infty_0(\Omega)^4$, in the sense that
  there is $d=d_{V,\Omega}>0$ such that
  \[
  \langle A\varphi,\varphi \rangle_{L^2(\Omega)^4}\ge 
  d\|\varphi\|_{L^2(\Omega)^4}^2,\quad \varphi\in  C^\infty_0(\Omega)^4.
  \]   
\end{prop}

\begin{proof}
Let $\varphi\in C^\infty_0(\Omega)^4$. Then we get
\begin{align*}
\langle A\varphi,\varphi \rangle_{L^2(\Omega)^4}&=\langle V^{-1}\mathcal{L}_0\varphi,\mathcal{L}_0\varphi \rangle_{L^2(\Omega)^4}\ge c_V\|\mathcal{L}_0\varphi\|_{L^2(\Omega)^4}^2\ge d\|\varphi\|_{L^2(\Omega)^4}^2.
\end{align*}
Here the last inequality follows from the estimate \cite[Theorem 10.3.7]{horbookII}
\[
\|\mathcal{L}_0\varphi\|_{L^2(\Omega)^4}\ge c_\Omega \|\varphi\|_{L^2(\Omega)^4},\quad \varphi\in  C^\infty_0(\Omega)^4,\quad c_\Omega>0.
\]

The claim follows. 

\end{proof}

\begin{prop}
The  second order operator $A$ is uniformly strongly elliptic in the sense that there is $c>0$ such that
\begin{equation}
\label{eq_ellipticity_D}
\langle \sigma(A)(x,\xi)\eta,\eta\rangle\ge c|\xi|^2|\eta|^2,\quad x\in \overline{\Omega}, \quad \xi\in \R^3,\quad \eta\in \C^4,
\end{equation}
where $\sigma(A)$ is the principal symbol of $A$.

\end{prop}

\begin{proof}
In view of homogeneity of \eqref{eq_ellipticity_D} it suffices to prove it
for $|\xi|=1$ and $|\eta|=1$.
Let 
$\eta=(\eta_+,\eta_-)^T\in \C^4$. Then since $V^{-1}$ is  an Hermitian positive-definite matrix valued function,
 we have
\begin{align*}
\langle \sigma(A)(x,\xi)\eta,\eta\rangle&=
\langle V^{-1}\begin{pmatrix} 0 & \sigma\cdot \xi\\
\sigma\cdot \xi & 0
\end{pmatrix}\eta,  \begin{pmatrix} 0 & \sigma\cdot \xi\\
\sigma\cdot \xi & 0
\end{pmatrix}\eta
\rangle\\
&\ge c_V\begin{vmatrix}
(\sigma\cdot\xi)\eta_-\\
(\sigma\cdot\xi)\eta_+
\end{vmatrix}^2_{\C^4}>0. 
\end{align*}
The latter inequality follows from the fact that
\[
\sigma\cdot\xi=\begin{pmatrix}
\xi_3 & \xi_1-i\xi_2\\
\xi_1+i\xi_2& -\xi_3
\end{pmatrix}, \quad
\det(\sigma\cdot\xi)=-|\xi|^2=-1.
\]
This proves \eqref{eq_ellipticity_D}.

\end{proof}

\begin{prop}
\label{prop_self_D}
The operator $A$, equipped with the domain
\[
\mathcal{D}(A)=H^1_0(\Omega)^4\cap H^2(\Omega)^4,
\]
is a positive self-adjoint operator on $L^2(\Omega)^4$.
\end{prop}

\begin{proof}
We shall consider the Friedrichs extension of $A$ on $C_0^\infty(\Omega)^4$, denoted also by $A$, which has the domain
\[
\mathcal{D}(A)=\mathcal{D}(Q)\cap \mathcal{D}(A_{\textrm{max}}).
\]
Here 
\[
Q(\varphi,\varphi)=\langle A\varphi,\varphi \rangle_{L^2(\Omega)^4}
\]
is the quadratic form associated with the operator $A$. The domain $\mathcal{D}(Q)$, also the form domain of $A$, is the 
completion of $C^\infty_0(\Omega)^4$ with respect to the norm $|\!|\!|\varphi|\!|\!|=\sqrt{Q(\varphi,\varphi)}$.
The maximal realization $A_{\textrm{max}}$ of the operator $A$ is defined by
\[
\mathcal{D}(A_{\textrm{max}})=\{u\in L^2(\Omega)^4:Au\in L^2(\Omega)^4\}.
\]

Let us now show that 
\begin{equation}
\label{eq_form_dom_D}
\mathcal{D}(Q)=H^1_0(\Omega)^4.
\end{equation}
Indeed, it is easy to see that  the norm $|\!|\!|\cdot|\!|\!|$  is equivalent to the following norm
\begin{equation}
\label{eq_norm_compl_D}
\|\mathcal{L}_0\cdot\|_{L^2(\Omega)^4}+ \|\cdot\|_{L^2(\Omega)^4}.
\end{equation}
Since 
$\|\mathcal{L}_0\cdot\|_{L^2(\Omega)^4}\le C\|\nabla\cdot\|_{L^2(\Omega)^2}$, we have 
\[
H^1_0(\Omega)^4\subset \mathcal{D}(Q). 
\]
On the other hand, it follows from \cite[Proposition 4.2]{NakTsu00} that the completion of $C^\infty(\overline{\Omega})^4$ with respect to the norm
\eqref{eq_norm_compl_D} is the space
\[
\mathcal{H}(\Omega)=\{u\in L^2(\Omega)^4:\mathcal{L}_0u\in L^2(\Omega)^4\}.
\]
Thus,
\[
\mathcal{D}(Q)\subset \mathcal{H}(\Omega). 
\]
It is shown in  \cite[Proposition 4.6]{NakTsu00} that the trace map
\[
\tau: C^\infty(\overline\Omega)^4\to C^\infty(\p \Omega)^4, \quad
u\mapsto u|_{\p \Omega},
\]
extends uniquely to a bounded map on $\mathcal{H}(\Omega)$. It follows from \cite[Proposition 4.6]{NakTsu00} that 
for any $u\in \mathcal{D}(Q)$, $\tau u=0$.  Hence, \cite[Proposition 4.10]{NakTsu00} implies that $u\in H^1_0(\Omega)^4$. 
This proves \eqref{eq_form_dom_D}. 

Hence,
\[
\mathcal{D}(A)=\{u\in H^1_0(\Omega)^4:Au\in L^2(\Omega)^4\}.
\]
As the operator $A$ is strongly elliptic and $\Omega$ has a smooth boundary, by elliptic regularity, see for instance
\cite[Section 7.5]{Grubbbook}, $\mathcal{D}(A)=H^1_0(\Omega)^4\cap H^2(\Omega)^4$. 

\end{proof}

It follows from Proposition \ref{prop_self_D} that   for any $\lambda\in \R$,
the operator $T_\lambda$, equipped with the domain
$\mathcal{D}(A)$
is a self-adjoint operator on $L^2(\Omega)^4$, and the form domain of $T_\lambda$ is $H^1_0(\Omega)^4$.

\begin{thm} The set of transmission eigenvalues for \eqref{eq_ITP_Helm_D} is discrete. 
\end{thm}

\begin{proof}

First note that the operator 
\[
-\lambda B+\lambda^2 C: \mathcal{D}(A)\to L^2(\Omega)
\]
is compact. Hence, the operator $T_\lambda:\mathcal{D}(A)\to L^2(\Omega)$ is Fredholm of index $0$, invertible at $\lambda=0$.  Thus, by analytic Fredholm theory, 
\[
T^{-1}_\lambda:L^2(\Omega)\to \mathcal{D}(A), \quad \lambda\in \C,
\]
is a meromorphic family of operators with residues of finite rank. This proves the claim. 

\end{proof}

As before, we see that the multiplicity of a transmission eigenvalue $\lambda\in \R$ is finite. 

\begin{thm} 
Let $V$ be a matrix-valued potential as above. Then there exists an infinite set of real transmission eigenvalues for \eqref{eq_ITP_Helm_D}. 
\end{thm}

\begin{proof}

First notice that Proposition
\ref{prop_helm_inf_n} continues to be valid for \eqref{eq_ITP_Helm_D}. 
It is therefore sufficient to prove the existence of transmission eigenvalues for the following problem, 
\begin{equation}
\label{eq_ITP_ball_D}
\begin{aligned}
(\mathcal{L}_0(D)-\lambda I_4)v=0 \quad \text{in}\quad B_\varepsilon(0),\\
(\mathcal{L}_0(D)-\lambda(1+\delta)I_4)w=0 \quad \text{in}\quad B_\varepsilon(0), \\
v-w\in H^1_0(B_\varepsilon(0))^4.
\end{aligned}
\end{equation}
Here  $\varepsilon>0$, $\delta>0$  and $B_\varepsilon(0)\subset\R^3$ is an open ball of radius $\varepsilon$ centered at the origin.

When considering \eqref{eq_ITP_ball_D},  we let $0\ne \lambda\in \R$ and study
\[
(\mathcal{L}_0(D)-\lambda I_4)v=0,
\]
where 
\[
v=\begin{pmatrix}v_+\\
v_-
\end{pmatrix},\quad 
v_\pm=\begin{pmatrix}
v_\pm^1\\
v_\pm^2
\end{pmatrix}.
\]
  Then we have
\begin{align*}
-\lambda v_++\sigma\cdot D v_-&=0,\\
\sigma\cdot D v_+-\lambda v_-&=0. 
\end{align*}
As 
\[
\sigma\cdot D=\begin{pmatrix}
D_3 & D_1-iD_2\\
D_1+iD_2 & -D_3
\end{pmatrix},
\]
we get
\begin{align*}
&v_+^1=\frac{1}{\lambda}(D_3v_-^1 + (D_1-iD_2)v_-^2)\\
&v_+^2=\frac{1}{\lambda}((D_1+iD_2)v_-^1 -D_3v_-^2),\\
&
\begin{pmatrix} 
-\Delta-\lambda^2 & 0\\
0 & -\Delta-\lambda^2
\end{pmatrix} \begin{pmatrix}
v_-^1\\
v_-^2
\end{pmatrix}=0.
\end{align*}
Considering the equation
\[
(\mathcal{L}_0(D)-\lambda(1+\delta)I_4)w=0,
\]
similarly, we obtain
\begin{align*}
&w_+^1=\frac{1}{\lambda(1+\delta)}(D_3w_-^1 + (D_1-iD_2)w_-^2)\\
&w_+^2=\frac{1}{\lambda(1+\delta)}((D_1+iD_2)w_-^1 -D_3w_-^2),\\
&
\begin{pmatrix} 
-\Delta-\lambda^2(1+\delta)^2 & 0\\
0 & -\Delta-\lambda^2(1+\delta)^2
\end{pmatrix} \begin{pmatrix}
w_-^1\\
w_-^2
\end{pmatrix}=0.
\end{align*}

Notice that to prove the existence of real transmission eigenvalues for the problem \eqref{eq_ITP_ball_D}, it suffices to restrict our attention to solutions $v,w$ of  
\eqref{eq_ITP_ball_D} such that $v_-^1=v_-^2=f(r)$ and $w_-^1=w_-^2=g(r)$, $r=|x|$, are spherically symmetric solutions 
of the following interior transmission problem,
\begin{equation}
\label{eq_transmis_bessel_D}
\begin{aligned}
(-\Delta-\lambda^2)f=0\quad \text{in}\quad B_\varepsilon(0),\\
(-\Delta-\lambda^2(1+\delta)^2)g=0 \quad \text{in}\quad B_\varepsilon(0),\\
f-g=0\quad \text{on}\quad \p B_\varepsilon(0),\\
\p_{r}f=\frac{\p_{r}g}{1+\delta}\quad \text{on}\quad \p B_\varepsilon(0).
\end{aligned}
\end{equation}
It is clear then that for such solutions, the boundary conditions
\begin{align*}
v_+^1=w_+^1\quad \text{on}\quad \p B_\varepsilon(0),\\
v_+^2=w_+^2\quad \text{on}\quad \p B_\varepsilon(0).
\end{align*}
are satisfied.

Since
\[
(\Delta+\lambda^2)f(r)=f''(r)+\frac{2}{r}f'(r)+\lambda^2f(r)=0,
\]
$f$ must be of the form 
\[
f(x)=c_0j_0(\lambda r),
\]
where $j_0$ is the spherical Bessel function of order zero and $c_0$ is a constant. 
In the same way, 
\[
g(x)=c_1j_0(\lambda(1+\delta) r),
\]
The boundary conditions in \eqref{eq_transmis_bessel_D} require that
\begin{align*}
&c_0j_0(\lambda\varepsilon)=c_1j_0(\lambda(1+\delta)\varepsilon),\\
& c_0j'_0(\lambda\varepsilon)=c_1j'_0(\lambda(1+\delta)\varepsilon).
\end{align*}
A nontrivial solution of this system exists if and only if
\begin{equation}
\label{eq_det_D}
\det\begin{pmatrix} j_0(\lambda\varepsilon) & -j_0(\lambda(1+\delta)\varepsilon)\\
 j'_0(\lambda\varepsilon) & -j'_0(\lambda(1+\delta)\varepsilon)
\end{pmatrix}=0.
\end{equation}
Since
\[
j_0(r)=\frac{\sin r}{r},\quad j'_0(r)=\frac{\cos r}{r}+\mathcal{O}(1/r^2),
\]
\eqref{eq_det_D} implies that
\begin{equation}
\label{eq_det_2_D}
\sin(\lambda\delta\varepsilon)+\mathcal{O}(1/\lambda)=0, \quad \lambda\to\infty.
\end{equation}
The existence of 
an infinite set of values $\lambda$ such that \eqref{eq_det_2_D} holds is clear as $\sin(\lambda\delta\varepsilon)$ is a periodic function taking positive and negative values. Each such $\lambda$ is a transmission eigenvalue for \eqref{eq_ITP_ball_D} and this completes the proof.

\end{proof}

\section{Generalized Rellich theorem}

In the last two sections, 
which do not depend on the material in Section 6 and 7,
we would like to explain the connection between interior transmission eigenvalues and scattering theory.
It is going to be provided by a generalization of the classical Rellich theorem,  proved in \cite{Hormander73}.

Let us start by summarizing the basic features of general scattering theory following  \cite[Chapter 14]{horbookII}.  Let $P_0$ be a partial differential operator  in $\R^n$ of order $m\ge 2$ with constant real coefficients,
\[
P_0(D)=\sum_{|\alpha|\le m} a_{\alpha}D^\alpha, \quad a_\alpha\in\R, \quad D_j=-i\frac{\partial}{\partial x_j},\quad j=1,\dots,n.
\] 
Assume that $\Lambda(P_0)=\{0\}$, see \eqref{eq_Lambda_p_0}, 
and that $P_0=P_0(D)$ is simply characteristic, i.e. 
\[
\sum_{|\alpha|\le m}|P_0^{(\alpha)}(\xi)|\le C(\sum_{|\alpha|\le 1}|P_0^{(\alpha)}(\xi)|+1),\quad C>0 .
\]
Examples are hypoelliptic operators \cite[Chapter 11]{horbookII} and operators of real principal type \cite[Chapter 8]{horbookI}.

In order to describe mapping properties of the boundary values of the $L^2$-resolvent of $P_0$, we follow  \cite[Chapter 14]{horbookII}
and introduce the following Banach spaces. Define
\begin{align*}
B&=\{v\in L^2(\R^n):\|v\|_B=\sum_{j=1}^\infty R_j^{1/2}(\int_{\Omega_j}|v|^2dx)^{1/2}<\infty \},\\
B^*&=\{u\in L^2_{\textrm{loc}}(\R^n):\|u\|_{B^*}=\sup_{j> 0} R_j^{-1/2}(\int_{\Omega_j}|u|^2dx)^{1/2}<\infty \},\\
\end{align*}
where
\begin{align*}
R_0=0,&\quad  R_j=2^{j-1}, \quad j=1,2,\dots,\\
\Omega_j=\{x\in \R^n:&R_{j-1}<|x|<R_j\}, \quad j=1,2,\dots.
\end{align*}
The space $B^*$ is the dual  of $B$ and we have
\[
L^2_\delta\subset B\subset L^2\subset B^*\subset L^2_{-\delta}, \quad \delta> 1/2. 
\]
The space of $C^\infty_0$-functions  is dense in $B$ but not in $B^*$. Its closure in $B^*$ is denoted
by $\stackrel{\circ}{B^*}$. Then $u\in L^2_{\textrm{loc}}(\R^n)$ belongs to $\stackrel{\circ}{B^*}$ if and only if
\[
\int_{|x|<R}|u|^2dx/R\to 0, \quad R\to\infty.
\]
We also define the Sobolev space version of $B^*$, associated to $P_0(D)$,
\[
B^*_{P_0}=\{u\in B^*:P^{(\alpha)}_0(D)u\in B^*, \forall \alpha\}.
\]

Let $Z(P_0)$ be the (necessarily finite) set of critical values of $P_0$, i.e.
\[
Z(P_0)=\{\lambda:\exists \xi\in \R^n \text{ s.t. } \nabla P_0(\xi)=0, P_0(\xi)=\lambda\}.
\] 
For $z\in \overline{\C^{\pm}}\setminus Z(P_0)$, the resolvent $R_0(z)=(P_0-zI)^{-1}$
of the simply characteristic operator  $P_0$ extends to a continuous map 
\[
 R_0(z):B\to B^*_{P_0}.
\]
Here $\overline{\C^{\pm}}=\{z:\pm \textrm{Im}\ z\ge 0\}$.  For $\lambda\in \R\setminus Z(P_0)$, 
the boundary values of the resolvent  are given by 
\[
R_0(\lambda\pm i0)f=\lim_{\varepsilon\to 0^+}F^{-1}((P_0-\lambda\mp \varepsilon i)^{-1}F(f)),
\]
where $F$ stands for  the Fourier transformation. 

We say that $u\in B^*$ is outgoing (incoming) if $u=R(\lambda+i0)f$ ($u=R(\lambda-i0)f$), $f\in B$ and $\lambda\in \R\setminus Z(P_0)$.
If $u$ is outgoing or incoming then $(P_0-\lambda)u=f$. 

For $\lambda\notin Z(P_0)$,  the level set $M_\lambda=\{\xi\in \R^n:P_0(\xi)=\lambda\}$ is an $(n-1)$-dimensional  $C^\infty$ submanifold  of $\R^n$.

It is known that $u$ is both outgoing and incoming, i.e. $u=R(\lambda+i0)f=R(\lambda-i0)f$, if and only if $F(f)=\hat f\equiv 0$ on $M_\lambda$.

We shall consider multiplicative perturbations of $P_0$ given by $V\in L^\infty(\R^n)$ with compact support. Such perturbations satisfy the short range condition introduced in \cite[Section 14.4]{horbookII}, i.e.  
\[
V:B^*_{P_0}\to B
\]
is compact.

In order to define the scattering amplitude, we recall the following fundamental result \cite[Theorem 14.6.8]{horbookII}.
\begin{thm} Assume that $\lambda\in \R\setminus Z(P_0)$. If $u\in B^*_{P_0}$ satisfies
\[
(P_0+V-\lambda)u=0,
\]
then
\begin{equation}
\label{eq_Lippmann-Schwinger}
u=u_{\pm}-R_0(\lambda\mp i0)Vu
\end{equation}
where 
\[
\hat u_{\pm}=v_\pm\delta(P_0-\lambda)=v_{\pm}\frac{dS}{|P'_0|} \quad\text{and} \quad v_\pm\in L^2(M_\lambda,dS).
\]
The map $v_-\mapsto v_+$ is a continuos bijection which extends to a unitary map
\[
\Sigma_\lambda: L^2(M_\lambda,\frac{dS}{|P'_0|})\to L^2(M_\lambda,\frac{dS}{|P'_0|}), \quad \Sigma_{\lambda}(v_-)=v_+.
\]
 \end{thm}
We call $v_-$ the incoming wave and $v_+$ the outgoing wave. The unitary map $\Sigma_\lambda$ is  the scattering matrix for the energy $\lambda$ and 
\[
A_{\lambda}=I-\Sigma_\lambda
\]
is  the scattering amplitude. 

The following result is well-known  and its proof is included for completeness only. 
\begin{lem} 
\label{lem_A_lambda} Assume that $\lambda\in \R\setminus Z(P_0)$ and  $u\in B^*_{P_0}$ satisfies the equation
$
(P_0+V-\lambda)u=0.
$
 Then the scattering amplitude $A_\lambda$ can be expressed through the Fourier transform of $Vu$ as follows,
\[
A_\lambda v_-(\xi)=2\pi i\hat{Vu}(\xi),\quad \xi\in M_\lambda.
\]
\end{lem}

\begin{proof}
By \cite[Theorem 14.6.8]{horbookII}, we get the following Lippmann-Schwinger equation for $u$, 
\[
u=u_{\pm}-R_0(\lambda\mp i0)Vu
\]
where 
\[
\hat u_{\pm}=v_\pm\delta(P_0-\lambda)=v_{\pm}\frac{dS}{|P'_0|} \quad\text{and} \quad v_\pm\in L^2(M_\lambda,dS).
\]
Thus,
\[
u_+-u_-=(R_0(\lambda- i0)-R_0(\lambda+ i0))Vu.
\]
Applying the Fourier transform, we have
\begin{align*}
(v_+-v_-)\frac{dS}{|P'_0|}&=\big(\frac{1}{P_0(\xi)-\lambda+i0}-\frac{1}{P_0(\xi)-\lambda-i0}\big)\hat{Vu}\\
&=2\pi i\delta(P_0(\xi)-\lambda)\hat{Vu}=2\pi i \hat{Vu}\frac{dS}{|P'_0|}
\end{align*}
that proves the claim.
\end{proof}

In order to describe the connection between transmission eigenvalues and the scattering amplitude, we shall now review the generalized Rellich theorem. Recall that a classical theorem of Rellich states that if  $v$ satisfies $(\Delta+k^2)v=0$ for $|x|>R_0$ and $v(x)|x|^{(n-1)/2}\to 0$ as $x\to\infty$, then $v(x)=0$ for $|x|\ge R_0$.

A far-reaching generalization of this result to broad classes of differential operators with constant coefficients has been given in 
\cite{Hormander73}.  Let us now state a much simplified version of \cite[Corollary 3.2]{Hormander73}. 
See also 
 \cite{AgmonHorm76}.

\begin{thm} \label{cor_Rellich}
Assume that $\lambda$ is not a critical value of $P_0(\xi)$, $\xi\in \R^n$, and that  there is a factorization  
\[
P_0(\zeta)-\lambda=cP_{1}^{m_1}(\zeta)\cdots P_{k}^{m_k}(\zeta), \quad c\in \R,
\]
for which every factor $P_i(\zeta)$ has real coefficients and is algebraically
 irreducible over $\C^n$. Assume furthermore that each $P_i(\zeta)$  has a non-empty set of real zeros. If 
$u\in \mathcal{S}'\cap L^2_{loc}$ is a solution of 
\begin{equation}
\label{eq_irreduce}
(P_0(D)-\lambda)u=f,
\end{equation}
 with $f\in L^2_{comp}$ and $u\in \stackrel{\circ}{B^*}$, then $u$ has a compact support and
\[
ch\supp (u)= ch \supp(f),
\]
where $ch$ stands for the convex hull.
\end{thm}

\begin{rem}

If $P_0(\zeta)-\lambda$ has an irreducible factor which has no simple real zero, then it was furthermore proved in \cite{Hormander73} that for any integer $N$ one can find $u\in L^\infty\cap C^\infty$ satisfying \eqref{eq_irreduce} with $f$ being compactly supported  and $u(x)=o(|x|^{-N})$ but $u$ not  compactly supported.
\end{rem}

To illustrate the main ideas involved in the proof of Theorem \ref{cor_Rellich}, for the convenience of the reader, we shall include a proof of the special case when $k=1$, $m=1$ and $c=1$. In doing so we shall follow  \cite{Hormander73} closely.

\begin{proof}

Set $P_\lambda(\zeta)=P_0(\zeta)-\lambda$, $\zeta\in \C^n$. 
The set of real zeros of the polynomial $P_\lambda(\xi)$ is equal to  $M_{\lambda}=\{\xi\in \R^n:P_\lambda(\xi)=0\}$. By the hypothesis of the theorem,   $M_{\lambda}\ne\emptyset$  is an $(n-1)$-dimensional submanifold of $\R^n$. 

Let $\xi_0\in M_\lambda$. Assume, as we may, that $\p_{\xi_n}P_0(\xi_0)\ne 0$, and write
$\xi=(\xi',\xi_n)\in \R^{n}$, $\xi'\in \R^{n-1}$.
By the implicit function theorem,  there is an analytic function $g:\R^{n-1}\to \R$, $g(\xi')=\sum a_\alpha(\xi'-\xi_0')^{\alpha}$,  defined locally near $\xi_0'$, such  that
\[
M_\lambda=\{(\xi',g(\xi')):\xi'\in \R^{n-1}\},
\]
locally near $\xi_0\in M_\lambda$.
Hence, the series $g(\zeta')=\sum a_\alpha(\zeta'-\xi_0')^{\alpha}$ also converges for $\zeta'\in \C^{n-1}$ near $\xi_0'$.

Let $M_\lambda^{\C}$ be the zero set of $P_\lambda$ in $\C^n$, i.e. 
\[
M_\lambda^{\C}=\{\zeta\in \C^n:P_\lambda(\zeta)=0\}.
\]
Then let us show that locally  near $\xi_0$, we have
\[
M_\lambda^{\C}=\{(\zeta',g(\zeta')):\zeta'\in \C^{n-1}\}.
\]
 Indeed,
$
P_\lambda(\xi',g(\xi'))=0
$ for any $\xi'\in \R^{n-1}$ near $\xi'_0$ and, hence, since $P_\lambda$ and $g$ are analytic, we get that
$P_\lambda(\zeta',g(\zeta'))=0
$ for any $\zeta'\in \C^{n-1}$ near $\xi'_0$.

Now since $u\in \stackrel{\circ}{B^*}$,  the Fourier transform $\hat f$ vanishes on  $M_\lambda$  (cf. \cite[Theorem 14.3.6]{horbookII}). Thus, by the analyticity of $\hat f$
and $g$, we have that 
$\hat f$ vanishes identically on an open neighborhood of $\xi_0$ in $M_\lambda^{\C}$.

Denote by
\[
A=\{\zeta\in \C^n:\hat f(\zeta)=0\}
\]
the set of the complex zeros of $\hat f$. 
Notice that the sets $A$ and $M_\lambda^{\C}$ are analytic. 
As we have assumed that $P_\lambda(\zeta)$ is algebraically irreducible over $\C^n$, the set 
$M_\lambda^{\C}$ is algebraically irreducible, and hence it is 
analytically irreducible as well, \cite{Tre60}. 
Moreover, $A\cap M_\lambda^{\C}$ contains an open neighborhood of $\xi_0$ in $M_\lambda^{\C}$. Then it follows from \cite[Corollary 2, Section 5.3]{Chirka} that 
\begin{equation}
\label{eq_zeros}
M_\lambda^{\C}\subset A.
\end{equation}
Thus, the function 
\[
\frac{\hat f(\zeta)}{P_\lambda(\zeta)}
\]
is  entire in $\C^n$. 
An application of \cite[Theorem 7.3.2]{horbookI} shows that there exists $v\in L^2_{comp}$ such that 
\[
P_\lambda(D)v=f.
\]
Hence, 
\[
P_\lambda(D)(u-v)=0, 
\]
and since $u\in \stackrel{\circ}{B^*}$, \cite[Theorem 7.1.27]{horbookI} implies that $u=v$. 
\end{proof}

\section{Injectivity of the scattering amplitude and transmission eigenvalues}

In this section we shall assume that the operator $P_0$ is hypoelliptic  and $\lambda\in \R$ is such that $P_0-\lambda$ satisfies the assumptions of Theorem \ref{cor_Rellich}. 
Let $V\in L^\infty(\R^n)$ be real-valued compactly supported in $\R^n$ with  $\supp(V)=\overline{\Omega}$, where $\Omega\subset \R^n$ is a bounded convex domain which is of class $C^\infty$, and $V\ge\delta>0$ a.e. in $\Omega$.

Consider the interior transmission problem,
\begin{equation}
\label{eq_ITP_2}
\begin{aligned}
(P_0-\lambda)v=0 \quad &\text{in} \quad \Omega,\\
(P_0+V-\lambda)w=0 \quad &\text{in} \quad \Omega,\\
 v-w \in H^{P_0}_0(\Omega).
\end{aligned}
\end{equation}

\begin{thm} There is a solution $(v,w)$ of \eqref{eq_ITP_2} with $0\ne v\in B^*$ if and only if the scattering amplitude $A_\lambda$ is not injective.
\end{thm}

\begin{proof}

 Assume that $A_\lambda$ is not injective. Then there exists a solution $u\in B^*_{P_0}$ of $(P_0+V-\lambda)u=0$ in $\R^n$ such that
$u=u_--u_s$ with  
$u_-\ne 0$ and 
the scattered wave 
\[
u_s:=R_0(\lambda+i0)Vu
\]  
being both 
 incoming and outgoing. By the mapping properties of the resolvent $R_0(\lambda+i0):B\to B^*_{P_0}$, we see that $u_-\in B^*_{P_0}$.  Now since the scattered wave $u_s$ is both 
 incoming and outgoing by Theorem \cite[Theorem 14.3.6]{horbookII} $u_s\in \stackrel{\circ}{B^*}$. Moreover,
 $(P_0-\lambda)u_s=Vu$. 
 Now Theorem \ref{cor_Rellich} implies that $u_s$ has compact support and $\supp(u_s)\subset\overline{\Omega}$, thanks to the convexity of $\Omega$. 
Since $u,u_s\in L^2_{\textrm{loc}}(\R^n)$ and $P_0$ is hypoelliptic, by hypoelliptic regularity \cite[Theorem 11.1.8]{horbookII}, 
 we have that $u_s\in B_{2,\tilde P_0}^{loc}(\R^n)$.
 As $\supp(u_s)\subset\overline{\Omega}$, a regularization argument shows that $u_s$ can be approximated by a sequence of $C_0^\infty(\Omega)$-functions, so that $u_s|_{\Omega}\in H^{P_0}_0(\Omega)$.
 Now setting $v=u_-|_\Omega\ne 0$ and $w=u|_\Omega$, we get a nontrivial solution to  \eqref{eq_ITP_2} .

Assume conversely that   the problem \eqref{eq_ITP_2} admits  a non-trivial solution $(v,w)$ with $0\ne v\in B^*$. Then by \cite[Theorem 14.3.3] {horbookII},  we have $\hat v=v_-dS$ with $v_-\in L^2(M_\lambda, dS)$.  Since $P_0$ is hypoelliptic, the  surface $M_\lambda$ is compact, and thus,  \cite[Theorem 14.3.8] {horbookII} implies  that $v\in B^*_{P_0}$.  As $v-w\in H^{P_0}_0(\Omega)$, we get that $w\in B^*_{P_0}$ and $v-w\in \stackrel{\circ}{B^*}$. 
Now $(P_0-\lambda)(v-w)=Vw$ and \cite[Theorem 14.3.6]{horbookII} yields that the Fourier transform $\hat{Vw}=0$ on $M_\lambda$. Hence, applying Lemma \ref{lem_A_lambda}, we get that $A_\lambda v_-=0$. 
\end{proof}

\begin{rem}
The convexity assumption on $\Omega$ can be removed if we require that $P_0$ is elliptic and $\R^n\setminus\overline{\Omega}$ is connected. 
\end{rem}

\section{Acknowledgements} We would like to thank David Colton and Fioralba Cakoni for providing us with some useful references. 
 The research of M.H. was partially supported by the NSF grant DMS-0653275 and he is grateful to the Department of Mathematics and Statistics at the University of Helsinki for the hospitality. The research of K.K. was financially supported by the
Academy of Finland (project 125599). 
The research of P.O. and L.P. was financially supported by Academy of Finland Center of Excellence programme 213476.

\end{document}